\documentclass[11pt]{article}

\usepackage[margin=1in]{geometry}
\usepackage{amsfonts}
\usepackage{amsmath,amsthm,amssymb}
\usepackage{bbm}
\usepackage{color}
\usepackage{graphicx}
\usepackage{pgf}
\usepackage{tikz}
\usetikzlibrary{arrows,automata,positioning}
\usepackage[latin1]{inputenc}
\allowdisplaybreaks

\newtheorem{lemma}{Lemma}[section]
\newtheorem{proposition}[lemma]{Proposition}

\newtheorem{corollary}[lemma]{Corollary}
\newtheorem{definition}[lemma]{Definition}
\newtheorem{example1}[lemma]{Example}
\newtheorem{rem1}[lemma]{Remark}
\newtheorem{assumption}[lemma]{Assumption}
\newtheorem{alg1}[lemma]{Algorithm}
\newtheorem{me1}[lemma]{Mechanism}

\newenvironment{remark}{\begin{rem1}\rm}{\end{rem1}}
\newenvironment{example}{\begin{example1}\rm}{\end{example1}}

\newenvironment{alg}{\begin{alg1}\rm}{\end{alg1}}

\newcommand{\bbr}{\mathbb{R}}

\newcommand{\bbn}{\mathbb{N}}

\newcommand{\bbt}{\mathbb{T}}

\newcommand{\ncal}{\mathcal{N}}

\newcommand{\xcal}{\mathcal{X}}

\newcommand{\ind}{\mathbbm{1}}

\begin{document}

\title{Impact of Contingent Payments on Systemic Risk in Financial Networks}
\author{Tathagata Banerjee \thanks{Washington University in St.\ Louis, Department of Electrical and Systems Engineering, St.\ Louis, MO 63130, USA.} \and
Zachary Feinstein \thanks{Washington University in St.\ Louis, Department of Electrical and Systems Engineering, St.\ Louis, MO 63130, USA. \tt{zfeinstein@wustl.edu}}}
\date{\today}
\maketitle
\abstract{
In this paper we study the implications of contingent payments on the clearing wealth in a network model of financial contagion. We consider an extension of the Eisenberg-Noe financial contagion model in which the nominal interbank obligations depend on the wealth of the firms in the network. We first consider the problem in a static framework and develop conditions for existence and uniqueness of solutions as long as no firm is speculating on the failure of other firms. In order to achieve existence and uniqueness under more general conditions, we introduce a dynamic framework. We demonstrate how this dynamic framework can be applied to problems that were ill-defined in the static framework.
}

\section{Introduction}\label{sec:intro}
The global financial crisis of 2007-2009 proved the need to study and understand how failures and losses spread through the financial system.  This effect, in which the distress of one bank puts the financial health of other banks in jeopardy, is called financial contagion.  
In an era of globalization and tight interconnections among the various financial entities, this type of contagion can spread rapidly causing a systemic crisis. Hence a thorough analysis of the different factors and mechanisms are of paramount importance.

However, the 2007-2009 financial crisis proved that not just banks, but insurance companies are also part of the financial system and hence linkages formed between banks and insurance companies can act as potential channels of financial contagion. These linkages are formed and resolved in a way that is different from normal bank loans. A typical example of such a linkage is a credit default swap [CDS]. 
A credit default swap is a contract in which a buyer pays a premium to a seller in order to protect itself against a potential loss due to the occurrence of a credit event that affects the value of the contract's underlying reference obligation, e.g., a corporate or sovereign bond.  The contract specifies the credit events that will trigger payment from the seller to the buyer.  Whereas such instruments can be used to hedge risks, they may also be used for speculative purposes to put a short position on the credit markets.  

The important role that such contingent linkages play is demonstrated by the financial crisis of 2007-2009.  As that crisis unfolded, AIG faced bankruptcy after the failure of Lehman Brothers due to the large payouts it was required to make on its CDS contracts referencing Lehman and mortgage backed securities. When the crisis hit, the sudden calls to pay out the CDS contracts put great pressure on AIG, which traditionally had a thin capital base. Consequently AIG had to be rescued by the U.S.\ Department of Treasury so as to avoid jeopardizing the financial health of firms which bought CDSs from AIG. However, despite the importance of these linkages, current models are unable to account for the conditional payments that an insurance or credit default swap contract would require.  We refer to \cite{BS12} for a preliminary study of the insurance and reinsurance market.

Our modeling follows the setting of the seminal paper of \cite{EN01}.  That work proposes a weighted graph to model the spread of defaults in the financial system. In this model, banks' liabilities are modeled through the edges. The banks use their liquid assets to pay off these liabilities; unpaid liabilities may cause other banks to default as well. Under simple conditions, they provide existence and uniqueness of the clearing payments and develop an algorithm for computing the same. The base model of \cite{EN01} has been extended in many directions to capture complexities in the financial system: Bankruptcy costs have been explored in \cite{E07,RV13,EGJ14,GY14,AW_15,CCY16}, cross-holdings in \cite{E07,EGJ14,AW_15}, extension to illiquid assets in \cite{CFS05,NYYA07,GK10,AFM13,CLY14,AW_15,AFM16,feinstein2015illiquid,feinstein2016leverage,feinstein2017multilayer}. Empirical studies on the spread of contagion has been done in \cite{ELS06,U11,CMS10,GY14}.

As far as the authors are aware, theoretical work on contingent payments and CDS in relation to systemic risk has not been explored much. \cite{SSB17,SSB16b} show that the clearing vector in the presence of generalized CDS contracts is not well-defined and need not exist. They further propose a static setting to model CDS payments and give sufficient conditions on the network topology for existence of a clearing solution. \cite{LPT17} considers such a model in a static framework and proposes a method to rewrite some classes of network topologies as an Eisenberg-Noe system. \cite{PCB14,CM2016} modeled CDS payments, but most of the reference entities are required to be external to the financial system.  \cite{minca2018reinsurance} modeled reinsurance networks and studied the implications of network topologies on existence and uniqueness of the liabilities and clearing payments. A different approach has been taken in \cite{HS12} in which a stochastic setting is used to analyze contagion caused by credit default swaps. The role of credit default swaps in causing financial contagion has been captured in several empirical studies, see e.g.\ \cite{PRY16,SGGS10}. 

The current work aims to provide a generalized theoretical framework in which to study credit default swaps and other contingent payments in the Eisenberg-Noe setting.  We focus on existence and uniqueness of the clearing payments under contingent payments without presupposing the nature of those payments or strong assumptions on the network topology.  This is in contrast to the aforementioned literature on CDS network models in which there is no guarantee that the realized networks would obey the required conditions. Hence it is paramount to develop results for a general network, irrespective of the topology. We do this by first considering the problem in a static framework where all claims are settled simultaneously.  In such a setting we find that uniqueness of the clearing solution follows so long as no firm is ``speculating'' on another firm's failure.
However, with speculation the problem in a static setting no longer satisfies the sufficient mathematical properties for uniqueness. In order to overcome this issue, we introduce a dynamic framework. This setting ensures both existence and uniqueness of a clearing solution under the usual conditions from \cite{EN01}. This dynamic framework is similar to the discrete time systems considered in \cite{BBF18,CC15}.

This work is organized in the following way: First, in Section~\ref{sec:setting}, we will introduce the mathematical and financial setting. In Section~\ref{sec:insurance}, we develop the static framework for incorporating contingent payments such as insurance and CDS, provide results on existence and develop conditions for uniqueness that are intimately related to considerations of insurance versus speculation. Further we demonstrate some shortcomings inherent to the static framework with contingent payments. In Section~\ref{sec:dynamic}, we introduce a discrete time dynamic framework and discuss existence and uniqueness results. Additionally we demonstrate how this framework can be applied to problems that were ill-defined in the static framework through numerical examples.

\section{Background}\label{sec:setting}
We begin with some simple notation that will be consistent for the entirety of this paper.  Let $x,y \in \bbr^n$ for some positive integer $n$, then 
\[x \wedge y = \left(\min(x_1,y_1),\min(x_2,y_2),\ldots,\min(x_n,y_n)\right)^\top,\] $x^- = -(x \wedge 0)$, and $x^+ = (-x)^-$.

Throughout this paper we will consider a network of $n$ financial institutions.  We will denote the set of all banks in the network by $\ncal := \{1,2,\ldots,n\}$.  Often we will consider an additional node $0$, which encompasses the entirety of the financial system outside of the $n$ banks; this node $0$ will also be referred to as society or the societal node.  The full set of institutions, including the societal node, is denoted by $\ncal_0 := \ncal \cup \{0\}$. We refer to \cite{feinstein2014measures,GY14} for further discussion of the meaning and concepts behind the societal node.

We will be extending the model from \cite{EN01} in this paper.  In that work, any bank $i \in \ncal$ may have obligations $L_{ij} \geq 0$ to any other firm or society $j \in \ncal_0$.  We will assume that no firm has any obligation to itself, i.e., $L_{ii} = 0$ for all firms $i \in \ncal$, and the society node has no liabilities at all, i.e., $L_{0j} = 0$ for all firms $j \in \ncal_0$.  Thus the \emph{total liabilities} for bank $i \in \ncal$ is given by $\bar p_i := \sum_{j \in \ncal_0} L_{ij} \geq 0$ and relative liabilities $\pi_{ij} := \frac{L_{ij}}{\bar p_i}$ if $\bar p_i > 0$ and arbitrary otherwise; for simplicity, in the case that $\bar p_i = 0$, we will let $\pi_{ij} = \frac{1}{n}$ for all $j \in \ncal_0 \backslash \{i\}$ and $\pi_{ii} = 0$.  
On the other side of the balance sheet, all firms are assumed to begin with some amount of assets $x_i \geq 0$ for all firms $i \in \ncal_0$.
The resultant \emph{clearing payments}, under a pro-rata payments environment, satisfy the fixed point problem
\begin{equation}\label{eq:EN-p}
p = \bar p \wedge \left(x + \Pi^\top p\right).
\end{equation}
That is, each bank pays the minimum of what it owes: $\bar p_i$ and what it has: $x_i + \sum_{j \in \ncal} \pi_{ji} p_j$.
The resultant vector of \emph{wealths} for all firms is given by
\begin{equation}\label{eq:equity}
V = x + \Pi^\top p - \bar p.
\end{equation}
Noting that these payments can be written as a simple function of the wealths ($p = [\bar p - V^-]^+$), we provide the following proposition.  We refer also to \cite{veraart2017distress,barucca2016valuation,BBF18} for similar notions of utilizing clearing wealth instead of clearing payments.
\begin{proposition}[Proposition 2.1 of \cite{BBF18}]\label{prop:EN-V}
$p \in [0,\bar p]$ is a clearing payment in the Eisenberg-Noe setting if and only if $p = [\bar p - V^-]^+$ for some $V \in \bbr^{n+1}$ satisfying the following fixed point problem
\begin{equation}\label{eq:EN-V}
V = x + \Pi^\top [\bar p - V^-]^+ - \bar p.
\end{equation}
Vice versa, $V \in \bbr^{n+1}$ is a clearing wealths if and only if $V$ is defined as in \eqref{eq:equity} for some clearing payments $p \in [0,\bar p]$ as defined in the fixed point problem \eqref{eq:EN-p}.
\end{proposition}
Due to the equivalence of the clearing payments and clearing wealths provided in Proposition~\ref{prop:EN-V}, we are able to consider the Eisenberg-Noe system as a fixed point of equity and losses rather than payments.

In \cite{EN01} results for the existence and uniqueness of the clearing payments (and thus for the clearing wealths as well) are provided.  In fact, it can be shown that there exists a unique clearing solution in the Eisenberg-Noe framework so long as $L_{i0} > 0$ for all firms $i \in \ncal$.  We will take advantage of this result later in this paper.  This is a reasonable assumption (as discussed in, e.g., \cite{GY14,feinstein2014measures}) as obligations to society include, e.g., deposits to the banks.

\begin{remark}\label{illiquid}
The analysis presented in this paper can be extended to include illiquid assets as discussed in, e.g., \cite{CFS05,AFM16,feinstein2015illiquid,feinstein2016leverage}, to include financial derivatives on illiquid assets.  This would allow for obligations to depend on the price of the illiquid assets, e.g., for hedging using put options.
\end{remark}

\section{Simultaneous network clearing with contingent payments}\label{sec:insurance}

\subsection{General setting}\label{sec:insurance-setting}
Let us now consider the case when the nominal liabilities between financial institutions depend explicitly on the wealths of the firms.  This is, for instance, the case with insurance, credit default swaps, reserve requirements with a central bank, or the default waterfall enacted by central counterparties; see Examples \ref{ex:insurance}-\ref{ex:CCP} for more details of those cases.  

As a general setting, this corresponds to the situation in which the nominal liabilities $L_{ij}: \bbr^{n+1} \to \bbr_+$ from bank $i \in \ncal$ to $j \in \ncal_0$ is a mapping from the vector of bank wealths into the obligations; as mentioned above, we will assume that $L_{0i} \equiv 0$ and $L_{ii} \equiv 0$ for all firms $i \in \ncal$.  That is, dependent on the actualized wealths $V \in \bbr^{n+1}$ of all institutions in the system, the nominal liabilities will adjust to be $L(V) \in \bbr^{(n+1) \times (n+1)}_+$ a nonnegative matrix with 0 diagonal.  In the case that the societal node is not desired, then this can be incorporated by setting $L_{i0} \equiv 0$ for all $i \in \ncal$. Thus we consider a static setting for these contingent payments, i.e., we assume all claims are resolved simultaneously and the nominal liabilities $L$ account for all layers of contingent claims.

\begin{assumption}\label{ass:L-cont}
The nominal liabilities $L_{ij}: \bbr^{n+1} \to \bbr_+$ are bounded with upper bound $\bar L_{ij} \geq 0$ for all institutions $i,j \in \ncal_0$. 
\end{assumption}

\begin{example}\label{ex:insurance}
Consider a static network model of external assets $x \in \bbr^{n+1}_+$ and liabilities $L^0 \in \bbr^{(n+1) \times (n+1)}_+$ (with corresponding total liabilities $\bar p^0$ and relative liabilities $\Pi^0$).  Firm $j \in \ncal_0$ purchased an \emph{insurance contract} from firm $i \in \ncal$ on the event that firm $k \in \ncal$ does not pay its obligations in full to firm $j$; this is encoded in the nominal liabilities function 
\[L_{ij}(V) = L_{ij}^0 + \sum_{k \in \ncal} \eta_{ij}^k(V) \frac{L_{kj}(V)}{\sum_{l \in \ncal_0} L_{kl}(V)} V_k^-.\]
In the above equation we set the parameter $\eta_{ij}^k: \bbr^{n+1} \to [0,1]$ to denote the level of insurance offered by the contract.  Logically we impose the condition that $\eta_{ij}^i \equiv 0$ for all firms $i \in \ncal$ and $j \in \ncal_0$ so as a firm is not insuring against itself.  We further impose a tree structure on the insurance, that is insurers will not directly insure nonpayments from other insurers in a cyclical manner.  This is codified in the condition that $\eta_{ij}^{k_1}\eta_{k_1j}^{k_2}\dots\eta_{k_mj}^i = 0$ for all $i,j,k_1,\cdots,k_m \in \ncal_0$.  This tree structure immediately implies the uniqueness of the nominal liabilities matrix $L: \bbr^{n+1} \to \bbr^{(n+1) \times (n+1)}_+$.  These conditions are related to the ``green core'' system in \cite{SSB17}.
In the case that $\eta_{ij}^k(V) > 1$, this is the situation of over-insurance which no longer is considered ``insurance'' in the strict legal sense; see Example~\ref{ex:CDS} for this more general setting.  More generally, over-insurance is implied by the condition $\sum_{i \in \ncal} \eta_{ij}^k(V) > 1$, i.e., the total amount of insurance on any payment should be bounded by 1. 
Though explained as a single insurance contract, multiple such contracts may be layered so that one financial institution may have insurance against the failures of multiple counterparties.
The simplest insurance contracts are such that $\eta_{ij}^k \equiv \hat\eta_{ij}^k \in [0,1]$, though by considering the functional we allow for situations in which insurance only pays losses exceeding a threshold $\tau_{ij}^k$, e.g.,
\[\eta_{ij}^k(V) = \hat\eta_{ij}^k\frac{\left[L_{kj}(V) - \frac{L_{kj}(V)}{\sum_{l \in \ncal_0} L_{kl}(V)} V_k^- + \tau_{ij}^k\right]^+}{\left[L_{kj}(V) - \frac{L_{kj}(V)}{\sum_{l \in \ncal_0} L_{kl}(V)} V_k^-\right]^+}.\]
Also within this framework we allow for reinsurance contracts; that is, insurance contracts that pay out once payments from an insurer reach a certain threshold so as to contain the losses for the insurer itself.
\end{example}
\begin{example}\label{ex:CDS}
As in Example \ref{ex:insurance}, consider an initial static network model with asset and liability parameters $(x,L^0)$.  Though similar to an insurance policy, a firm may purchase credit default swaps.  Firm $j \in \ncal_0$ purchased a \emph{credit default swap [CDS]} from firm $i \in \ncal$ on the failure of firm $k \in \ncal$ is encoded in the formula
\[L_{ij}(V) = L_{ij}^0 + \sum_{k \in \ncal} \eta_{ij}^k(V)V_k^-.\]
In this example we define $\eta_{ij}^k: \bbr^{n+1} \to \bbr_+$ without restriction on the number of swaps purchased or the existence of an insurable interest.
In such a way we allow for so-called ``naked'' CDSs where the payments to firm $j$ are not based on any insurable interest in firm $k$.
\end{example}
\begin{example}\label{ex:stability}
As in Example \ref{ex:insurance}, consider an initial static network model with asset and liability parameters $(x,L^0)$.  We will now consider a system in which all firms must pay towards a centralized stability fund.  That is, prior to the start some amount $y \in [0,x]$ of the external assets are provided from each firm used in the stability fund.  In the case of failures this fund would support the defaulting firms.  Consider this centralized fund to be denoted as node $B$ and let $\ncal_B = \ncal_0 \cup \{B\}$.  This system can be described in which the bailout fund is capitalized prior to clearing or as part of clearing.  If the bailout is collected prior to clearing than this system is described by external assets of $x_i - y_i \geq 0$ for all firms $i \in \ncal$ and $\sum_{i \in \ncal} y_i \geq 0$ for the stability fund node $B$ and liabilities of
\[L_{ij}(V) = L_{ij}^0 \; \forall i,j \in \ncal_0, \qquad L_{Bi}(V) = V_i^- \; \forall i \in \ncal, \qquad L_{iB}(V) = 0 \; \forall i \in \ncal.\]
The payments to this stability fund can also be made as a part of clearing.  In this case the external assets are $x$ and liabilities are
\[L_{ij}(V) = L_{ij}^0 \; \forall i,j \in \ncal_0, \qquad L_{Bi}(V) = V_i^- \; \forall i \in \ncal, \qquad L_{iB}(V) = y_i \; \forall i \in \ncal.\]
This can be extended further by setting the payments to the stability fund $y$ to itself be a function of the wealth of each institution.  This allows for concepts such as pooled reserve requirements to be encoded into our general framework.
\end{example}
\begin{example}\label{ex:CCP}
The final general conceptual example we wish to present is the situation of introducing a \emph{central counterparty [CCP]}.  In this setting, the network topology follows a star shape, i.e., firms only have liabilities to and from some centralized CCP node.  The true CCP rules, however, also include what is called a default waterfall.  The default waterfall kicks in when the CCP is unable to pay out in full through the initial collected liabilities and margin payments.  In such a case the remaining solvent firms are forced to provide more liquidity to the CCP node.  In a broad sense, this fits within the general framework considered herein as the obligations to the CCP are directly dependent on the wealths of all firms in the system.  CCPs are described in more detail in \cite{AB2014ccp,Murphy2012ccp,cont2015ccp,CM2016}.
\end{example}

As in the construction of the Eisenberg-Noe setting \cite{EN01}, the total and relative liabilities will implicitly be functions of the system wealths as well, i.e.,
\begin{align}
\label{eq:pbar} \bar p_i(V) &= \sum_{j \in \ncal_0} L_{ij}(V)\\
\label{eq:a} \pi_{ij}(V) &= \begin{cases}\frac{L_{ij}(V)}{\bar p_i(V)} &\text{if } \bar p_i(V) > 0\\ \frac{1}{n} &\text{if } \bar p_i(V) = 0, \; i \neq j\\ 0 & \text{if } \bar p_i(V) = 0, \; i = j\end{cases}
\end{align}
for firms $i,j \in \ncal_0$ and system equities $V \in \bbr^{n+1}$.

With this contingent setting we can define the extension of the Eisenberg-Noe framework as the fixed point problem
\begin{equation}\label{eq:EN-insurance}
V = x + \Pi(V)^\top [\bar p(V) - V^-]^+ - \bar p(V).
\end{equation}
That is, the wealths are the sum of external assets and payments from other banks minus the payments owed.  This could equivalently be defined directly as the payments as is done in \cite{EN01}, we choose to consider the wealths directly in this work as it is easier to consider the examples, e.g., insurance payments.  The realized payments can be defined (as discussed previously without contingent payments) by $p = [\bar p(V) - V^-]^+$.

\begin{proposition}\label{prop:V-bounded}
Under Assumption \ref{ass:L-cont}, any fixed point wealth $V \in \bbr^{n+1}$ of \eqref{eq:EN-insurance} lies within the compact set $\prod_{i = 1}^n [x_i-\sum_{j \in \ncal_0} \bar L_{ij} , x_i + \sum_{j \in \ncal} \bar L_{ji}]$. 
\end{proposition}
\begin{proof}
The result is immediate by the boundedness properties of Assumption \ref{ass:L-cont}.
\end{proof}

\begin{corollary}\label{cor:exist-cont}
Under Assumption~\ref{ass:L-cont}, there exists an equilibrium wealth of \eqref{eq:EN-insurance} if $L_{ij}: \bbr^{n+1} \to \bbr_+$ is continuous as a function of wealths for all firms $i,j \in \ncal_0$.
\end{corollary}
\begin{proof}
This follows from the compactness argument of Proposition~\ref{prop:V-bounded} and the Brouwer fixed point theorem.
\end{proof}

Though in Corollary~\ref{cor:exist-cont} we have proven the existence of an equilibrium solution to \eqref{eq:EN-insurance}, this need not be a unique solution.  The following example illustrates a simple network with multiple equilibria.  Further, Corollary~\ref{cor:exist-cont} and \eqref{eq:EN-insurance} implicitly assume that there are no bankruptcy costs.  With such costs (as introduced in \cite{RV13}), Corollary~\ref{cor:exist-cont} will no longer apply.  See also Remark~\ref{rem:bankruptcy} for a discussion on sufficient conditions to guarantee existence of a clearing wealths vector under bankruptcy costs.

\begin{example}\label{ex:nonunique}
Consider the network with $n = 3$ banks, and \emph{without} the societal node. This network is depicted in Figure \ref{fig:nonunique}. Banks 1 and 3 have $x_1 = x_3 = 0$ external assets and bank 2 begins with $x_2 = 3/16$ external assets.  We consider the case in which $L_{23} = L_{32} \equiv 1$ are fixed obligations whereas the first bank has purchased a credit default swap on the third institution defaulting on its obligations from the second institution that pays out $L_{21}(V) = V_3^-$.  No other exposures exist within this system.  The system of wealths must therefore satisfy
\begin{align*}
V_1 &= \frac{V_3^-}{1+V_3^-}(1 + V_3^- - V_2^-)^+\\
V_2 &= \frac{3}{16} + (1-V_3^-)^+ - (1+V_3^-)\\ 
V_3 &= \frac{1}{1+V_3^-}(1 + V_3^- - V_2^-)^+ - 1.
\end{align*}
It can be shown that the following are both equilibrium wealths of the contingent network:
\begin{itemize}
\item $V = \left(0 \; , \; 3/16 \; , \; 0\right)^\top$, i.e., payments are given by $p = \bar p(V) - V^- = \left(0 \; , \; 1 \; , \; 1\right)^\top$, and
\item $V = \left(3/16 \; , \; -21/16 \; , \; -3/4\right)^\top$, i.e., payments are given by $p = \bar p(V) - V^- = \left(0 \; , \; 7/16 \; , \; 1/4\right)^\top$.
\end{itemize}
\begin{figure}[h!]
\centering
\begin{tikzpicture}
\tikzset{node style/.style={state, minimum width=0.36in, line width=0.5mm, align=center}}
\node[node style] at (0,0) (x1) {$x_1 = 0$};
\node[node style] at (5,0) (x2) {$x_2 = \frac{3}{16}$};
\node[node style] at (10,0) (x3) {$x_3 = 0$};

\draw[every loop, auto=right, line width=0.5mm, >=latex]
(x2) edge[dashed] node {$L_{21}(V) = V_3^-$} (x1)
(x2) edge[bend right=20] node {$L_{23} \equiv 1$} (x3)
(x3) edge[bend right=20] node {$L_{32} \equiv 1$} (x2);
\end{tikzpicture}
\caption{Example~\ref{ex:nonunique}: A graphical representation of the network model with 3 banks which accepts more than one clearing solution.}
\label{fig:nonunique}
\end{figure}
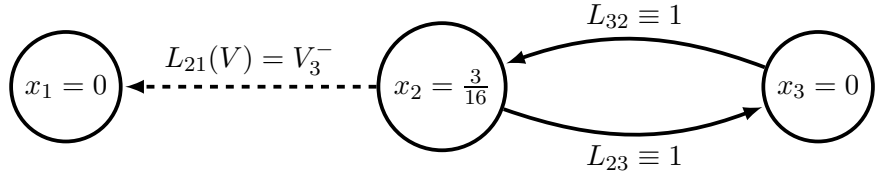
\end{example}

\subsection{A nonspeculative financial network}\label{sec:insurance-nonspeculative}
We will now impose additional properties upon the financial system to have stronger existence results, culminating in uniqueness of the clearing solutions.  These results provide monotonicity of the wealth of the banks in the financial system.  The first of such properties, defined as a nonspeculative property, is provided below in Definition \ref{defn:nonspeculative}.

\begin{definition}\label{defn:nonspeculative}
Firm $i \in \ncal_0$ is called \emph{\textbf{nonspeculative}} if 
\[x_i + \sum_{j \in \ncal} \pi_{ji}(V)[\bar p_j(V) - V_j^-]^+ - \bar p_i(V)\]
is nondecreasing in $V \in \bbr^{n+1}$.  The network $\ncal_0$ is called nonspeculative if all firms $i \in \ncal_0$ are nonspeculative.
\end{definition}

We call the property in Definition \ref{defn:nonspeculative} ``nonspeculative'' as it provides conditions so that firm $i \in \ncal_0$ does not benefit from (i.e., speculate on) the failure of another firm.  We do, however, allow for firm $i$ to \emph{hedge} its exposure to other firms.  This exposure can be either direct or indirect.

\begin{remark}\label{rem:green-core}
The nonspeculative framework considered herein is similar to, and can be considered as an extension of the properties considered in \cite{SSB17}.  In that work, the monotonicity property, in the definition of nonspeculative systems that we consider, is specified for credit default swaps.  Properties on solutions, which we will derive from this nonspeculative property, are considered as a function of the network topology in \cite{SSB17}; in fact, the topological features required in \cite{SSB17} guarantee that the ``green core'' system is inherently nonspeculative.
\end{remark}

\begin{lemma}\label{lemma:nonspeculative}
Under Assumption~\ref{ass:L-cont}, any nonspeculative system has a greatest and least equilibrium wealth $V^\uparrow \geq V^\downarrow$ satisfying \eqref{eq:EN-insurance} existing within the compact space $\prod_{i \in \ncal_0} [x_i - \sum_{j \in \ncal_0} \bar L_{ij} , x_i + \sum_{j \in \ncal} \bar L_{ji}]$.  Additionally, under all clearing vectors the value of the equity of each node of the financial system is the same, that is, if $V$ and $\hat V$ are any two clearing wealths then $V^+ = \hat V^+$.
\end{lemma}
\begin{proof}
By the nonspeculative property we can apply the Tarski fixed point theorem to get the existence of a maximal and minimal fixed point $V^\uparrow \geq V^\downarrow$.  By Proposition \ref{prop:V-bounded} we have that such solutions must exist within the provided compact space.

Now we will show the uniqueness of the positive equities by proving that $(V^\uparrow)^+ = (V^\downarrow)^+$.  By definition we know that $(V^\uparrow)^+ \geq (V^\downarrow)^+$, so as in \cite[Theorem 1]{EN01} we will prove that the total positive equity in the system remains constant.  Let $V \in \bbr^{n+1}$ be some equilibrium wealth solution, then since $\bar p_0 \equiv 0$ by definition we recover that
\begin{align*}
\sum_{i \in \ncal_0} V_i^+ &= \sum_{i \in \ncal_0} \left(x_i + \sum_{j \in \ncal} \pi_{ji}(V) [\bar p_j(V) - V_j^-]^+ - \bar p_i(V)\right)^+\\
&= \sum_{i \in \ncal_0} \left(x_i + \sum_{j \in \ncal} \pi_{ji}(V) [\bar p_j(V) - V_j^-]^+ - [\bar p_i(V) - V_i^-]^+\right)\\
&= \sum_{i \in \ncal_0} x_i + \sum_{j \in \ncal} [\bar p_j(V) - V_j^-]^+ \sum_{i \in \ncal_0} \pi_{ji}(V) - \sum_{i \in \ncal_0} [\bar p_i(V) - V_i^-]^+\\
&= \sum_{i \in \ncal_0} x_i + \sum_{j \in \ncal} [\bar p_j(V) - V_j^-]^+ - \sum_{i \in \ncal_0} [\bar p_i(V) - V_i^-]^+\\
&= \sum_{i \in \ncal_0} x_i.
\end{align*}
Therefore $\sum_{i \in \ncal_0} (V_i^\uparrow)^+ = \sum_{i \in \ncal_0} (V_i^\downarrow)^+$ and thus $(V^\uparrow)^+ = (V^\downarrow)^+$.
\end{proof}

\begin{remark}\label{rem:bankruptcy}
The inclusion of bankruptcy costs to this setting, in much the same way as accomplished in \cite{RV13,SSB17}, would guarantee the existence of a maximal and minimal clearing wealths vector under the assumptions of Lemma~\ref{lemma:nonspeculative}.  Much as in \cite{SSB17}, without the nonspeculative assumption, existence of a solution may not exist since Corollary~\ref{cor:exist-cont} will no longer apply.
\end{remark}

We will now give additional properties for the societal node $0$ to satisfy.
\begin{assumption}\label{ass:society}
All firms $i \in \ncal$ have strictly positive obligations to society, i.e., $L_{i0}: \bbr^{n+1} \to \bbr_{++}$.  Additionally, the obligations to society depend only on the negative wealths of all firms, i.e., $L_{i0}(V) = L_{i0}(-V^-)$ for all $i \in \ncal$.
\end{assumption}

\begin{definition}\label{defn:strictly-nonspeculative}
The societal node is called \emph{\textbf{strictly nonspeculative}} if 
\[\sum_{j \in \ncal} \pi_{j0}(V)[\bar p_j(V) - V_j^-]^+\]
is strictly increasing in $V \in \bbr^{n+1}_-$.  The network $\ncal_0$ is called strictly nonspeculative if the system  is nonspeculative and the societal node is strictly nonspeculative.
\end{definition}
We call this property strictly nonspeculative since it provides the condition that society does strictly worse as any firm defaults by any additional amount.  This requires that society can never perfectly hedge its risk, and as a consequence is strictly not speculating on any firm's failure.  This is a reasonable property as society should always be exposed to banking failures to some degree through, e.g., deposits and the payments necessary for deposit insurance.

\begin{corollary}\label{cor:nonspeculative}
Under Assumptions~\ref{ass:L-cont} and \ref{ass:society}, any strictly nonspeculative system has a unique equilibrium wealth of \eqref{eq:EN-insurance} existing within the compact space $\prod_{i \in \ncal_0} [x_i - \sum_{j \in \ncal_0} \bar L_{ij} , x_i + \sum_{j \in \ncal} \bar L_{ji}]$. 
\end{corollary}
\begin{proof}
Using Lemma \ref{lemma:nonspeculative} we have the existence of greatest and least fixed points $V^\uparrow \geq V^\downarrow$.  Let us assume there exists some firm $i \in \ncal$ such that $0 > V_i^\uparrow > V_i^\downarrow$ (otherwise uniqueness is guaranteed by nonexistence of such a firm as well as the uniqueness of the positive equities).  By the definition of the equilibrium wealth of the society node
\begin{align*}
V_0^\uparrow &= \sum_{j \in \ncal} \pi_{j0}(V^\uparrow) [\bar p_j(V^\uparrow) - (V_j^\uparrow)^-]^+\\
&> \sum_{j \in \ncal} \pi_{j0}(V^\downarrow) [\bar p_j(V^\downarrow) - (V_j^\downarrow)^-]^+ = V_0^\downarrow.
\end{align*}
However, immediately we know that the societal node has positive equity, therefore by Lemma \ref{lemma:nonspeculative} it must follow that $V_0^\uparrow = V_0^\downarrow$, which is a contradiction so uniqueness must follow.
\end{proof}

We will now provide a version of the \emph{fictitious default algorithm} (as discussed in, e.g., \cite{EN01,RV13,AFM16,feinstein2015illiquid,AW_15}) for the contingent payments described in \eqref{eq:EN-insurance}.  Algorithm~\ref{alg:fda} provides the maximal fixed point $V^\uparrow$ under the conditions of Lemma~\ref{lemma:nonspeculative}, that is, for a network of nonspeculative banks.  It can easily be modified to provide the minimal fixed point $V^\downarrow$ instead.  Under the assumptions of Corollary \ref{cor:nonspeculative} this algorithm results in the unique equilibrium wealths.

\begin{alg}\label{alg:fda}
Consider the setting of Lemma~\ref{lemma:nonspeculative} such that $L(V) = L(-V^-)$ for every $V \in \bbr^{n+1}$.  The greatest clearing wealths $V^\uparrow$ can be found in at most $n$ iterations of the following algorithm.
Initialize $k=0$, $D^0 = \emptyset$, and $V^0 = x + \Pi(0)^\top \bar p(0) - \bar p(0)$.
\begin{enumerate}
\item \label{FDA:2} Increment $k=k+1$;
\item Denote the set of insolvent banks by
$D^k=\{i \in \ncal \; | \; V_i^{k-1}<0\}$;
\item \label{alg:fda-terminate} If $D^k=D^{k-1}$ then terminate;
\item Define the matrix  $\Lambda \in \{0,1\}^{(n+1) \times (n+1)}$ so that \[\Lambda_{ij} = \begin{cases} 1 &\text{if } i = j \in D^k \\ 0 &\text{else}\end{cases};\]
\item \label{alg:fda-fixedpt} $V^k = \hat V$ is the maximal solution of the following fixed point problem:
\begin{align*}
\hat V = x + \Pi(\Lambda \hat V)^\top [\bar p(\Lambda \hat V)+\Lambda \hat V]^+ - \bar p(\Lambda \hat V)
\end{align*}
in the domain $\prod_{i \in \ncal_0} [x_i - \sum_{j \in \ncal_0} \bar L_{ij} , x_i + \sum_{j \in \ncal} \bar L_{ji}]$;
\item Go back to step \ref{FDA:2}.
\end{enumerate}
\end{alg}

In Algorithm~\ref{alg:fda} at most $n$ iterations are needed, as opposed to $n+1$ as would generally be stated for the fictitious default algorithm of \cite{EN01}.  This is due to the fact that, by definition, the societal node $0$ has no obligations and therefore cannot default.  The additional condition that $L(V) = L(-V^-)$ for Algorithm~\ref{alg:fda} corresponds to the case in which the nominal liabilities only depend on the set of solvent institutions and the shortfall of each insolvent institution.  This is satisfied for, e.g., CDSs as described in Example~\ref{ex:CDS}.

We now consider a simple sensitivity analysis of the clearing wealths $V$ under uncertainty in the initial endowments $x \in \bbr^{n+1}_+$. In so doing we are able to consider similar comparative statics results as in Lemma 5 of \cite{EN01}. 
\begin{proposition}\label{prop:continuous}
Consider the setting of Corollary~\ref{cor:nonspeculative} such that the nominal liabilities $L_{ij}: \bbr^{n+1} \to \bbr_+$ are continuous for every pair of firms $i,j \in \ncal_0$.  The unique clearing wealths $V: \bbr^{n+1}_+ \to \bbr^{n+1}$ are continuous and nondecreasing as a function of the initial endowments $x \in \bbr^{n+1}_+$.
\end{proposition}
\begin{proof}
The proof is presented in the appendix.
\end{proof}

\subsection{Shortcomings}\label{sec:disadvantages}
We wish to conclude our discussion of the system with contingent obligations under simultaneous claims by considering shortfalls to this approach.  In the subsequent section we will consider a dynamic approach to overcome these shortcomings, as well the issues on existence and uniqueness considered above and by \cite{SSB17,SSB16b}.

Consider a setting in which a firm takes out an insurance contract on its own failure.  This setting is of particular interest as it is inherently the type of contingent payment owed to a central counterparty as part of the default waterfall in the CCP framework considered in Example~\ref{ex:CCP} or the stability fund discussed in Example~\ref{ex:stability}.  To further simplify this setting, and again to make it relevant with regards to the CCP framework, we will consider the case in which the firm(s) offering insurance have enough assets to make all payments in full.  Allow the firm taking out the insurance contracts to be firm $1$.  If the sum total of all contingent payments are exactly enough to make firm $1$ whole again, i.e., all contingent payments to firm $1$ sum up to $V_1^-$, then in principle firm $1$ will never default.  However, in equilibrium, this is not what the insurance payments will be; in fact, if the insurance payments in a fixed point were to make firm $1$ whole then no insurance payments would be made and the initial shortfall would be realized once more.  Therefore, in equilibrium, it must be the case that firm $1$ will default even if they are paid the insurance.  As a simple demonstrative example, if firm $1$ only has obligations to the societal node (allowing for us to ignore all feedback effects from firm $1$ paying more and having a higher recovery rate through the network) then the insurance will add up to exactly \emph{half} of firm $1$'s initial shortfall in the fixed point as the new shortfall for firm $1$ will be equal to the contingent payments that are being made.  

However, while this conceptual problem with a firm taking out an insurance contract on its own losses is important, it can conceivably be overcome by providing a sufficiently complicated structure to the contingent payments.  A more subtle, but pernicious, flaw is that this contingent payment system is speculative by construction.  Namely, if the wealth of firm $1$ is lowered, no other firm does better (firm $1$ will pay out less and the insurance companies will have higher claims to pay), but firm $1$ itself improves its wealth.  This is due to the nonspeculative property being constructed in which firm $1$ does not directly get hit by its own lower wealth, but would only occur in network effects that would be on the second order, not in evidence in the single iteration of the definition.  Thus, even though the network is constructed from the notion that no firm benefits in the case of defaults (and this would be evidenced in any equilibrium), the monotonicity of the nonspeculative property is a stronger construct that cannot be satisfied so easily from a conceptual standpoint.

The above described problems could, in specific circumstances (e.g., a single insurer and only one contingent payment contract or a ``green core'' system from \cite{SSB17}), be overcome by reformulating the payments appropriately.  However, in the general case with each contract incorporating no speculation from a financial perspective, this system would have the aforementioned shortcomings.  These challenges, along with the inability to deal with speculative systems in general, stem from structural issues in such a static framework.  Specifically, insurance, and contingent payments more generally, are paid on specific claims, not simultaneous to the claim being made. This necessitates a dynamic approach to this problem, which we will discuss in the subsequent section.

\section{Dynamic Framework}\label{sec:dynamic}
As detailed above, the static, simultaneous claims, model presented has both mathematical and economic issues that the authors are not aware of any way to overcome in a general setting.  These problems are associated with the presence of, potentially, infinite cycles.  Much as with \cite{KV16}, these cycles could alternate between two states, particularly for speculative systems.  That is, for instance, insurance is paid out because a bank is insolvent, but because of this insurance payment the firm is no longer insolvent and no payment would be necessary.
As in \cite{KV16}, we will consider an algorithmic approach to this issue.  We thus propose a simple dynamic framework.  Additionally, we consider this setting to be more realistic than the static setting considered above and by \cite{SSB17} as the financial system does not include the payment of, e.g., a CDS on the obligation inherent in that contract.

\subsection{General Setting}
We adapt the framework introduced in \cite{BBF18} for the purposes of constructing a simple dynamic framework for contingent payments. Consider a discrete set of clearing times $\bbt$, e.g., $\bbt = \{0,1,\dots,T\}$ for some (finite) terminal time $T < \infty$ or $\bbt = \bbn$.
For processes we will use the notation from \cite{cont2013ito} such that the process $Z: \bbt \to \bbr^n$ has value of $Z(t)$ at time $t \in \bbt$ and history $Z_t := (Z(s))_{s = 0}^t$.
As an explicit extension to \cite{BBF18}, we consider the external (incoming) cash flow $x:\bbt \times \bbr^{(n+1) \times |\bbt|} \to \bbr^{n+1}_+$ and nominal liabilities $L: \bbt \times \bbr^{(n+1) \times |\bbt|} \to \bbr^{(n+1) \times (n+1)}_+$ to be functions of the clearing time \emph{and} prior wealths.  For simplicity, we will consider $x(t,\cdot) := x(t)$ to be independent of the prior wealths, though it may still depend on time.  The distinguishing feature of this model compared to the static Eisenberg-Noe model (or the static contingent payment model above and in \cite{SSB17}) is that the system parameters may depend on prior times.  For example, if firm $i$ has positive equity at time $t-1$ (i.e., $V_i(t-1) > 0$) then these surplus assets are available to firm $i$ at time $t$ in order to satisfy its obligations.  In the contingent setting, the wealths of all banks at time $t-1$ may affect the obligations due at time $t$ as well.

We note that in \cite{BBF18} all unpaid obligations from a prior time are assumed to roll forward automatically. That is, if firm $i$ has negative wealth at time $t-1$ then the debts that the firm has not yet paid will roll forward in time and be due at the next time point. Under such an assumption, no firm is deemed to default on its obligations until the terminal time $T$.  Herein, with the explicit consideration of the contingent payments, we may ``zero out'' a firm before the terminal date if it is deemed to default in much the same as in, e.g., \cite{CC15}.  While we can incorporate the notion of loans from \cite{CC15} as well, we will restrict our analysis to debts rolling forward in time so as to simplify the discussion.

As noted above, in addition to the structure from \cite{BBF18}, the nominal liabilities will explicitly depend on the clearing wealths of the prior time(s), i.e., $L: \bbt \times \bbr^{(n+1) \times |\bbt|} \to \bbr^{(n+1) \times (n+1)}_+$.  Often, to make this difference explicit especially in examples, we consider the full nominal liabilities $L$ to be a combination of two components: a non-contingent component $L^0: \bbt \to \bbr^{(n+1) \times (n+1)}_+$ which is only a function of clearing times and a contingent component $L^c: \bbt \times \bbr^{(n+1) \times |\bbt|} \to \bbr^{(n+1) \times (n+1)}_+$ which is a function of both the clearing times and the past history (but only encodes the contingent payments based on the past history).  That is, $L = L^0 + L^c$.

As a descriptive consideration of the contingent obligations $L^c$, consider the insurance-based (Example~\ref{ex:insurance}) or credit default swap (Example~\ref{ex:CDS}) scenarios of the previous section.  For instance, a bank $j$ may purchase a credit default swap from bank $i$ on the failure of firm $k$ as described in Example~\ref{ex:CDS}.  As opposed to the simultaneous claims setting in Section~\ref{sec:insurance}, in this dynamic setting we consider an order of operations.  That is, first firm $k$ must fail at time $t-1$, and only then would the credit default swap be paid at time $t$.  This delay in payments is a reflection of the real financial system in which there is a time between a claim being made by bank $j$ to $i$ and the payment on that claim.  The payment due to this credit default swap would be incorporated in $L^c_{ij}$ but not $L^0_{ij}$.

Even with this important distinction, we can use the same methodology as in \cite{BBF18} to prove existence and uniqueness of the clearing wealths in this setting.  The following assumption, with the concept taken from \cite{BBF18} guarantees that all firms are solvent at the start of the system and that the system is a regular network as described by \cite{EN01}.
\begin{assumption}\label{ass:initial}
Before the time of interest, all firms are solvent and liquid.  That is, $V_i(-1) \geq 0$ for all firms $i \in \ncal_0$.  
Additionally, all firms have positive external cash flow or obligations to society at all times $t \in \bbt$, i.e., $x_i(t) + L_{i0}(t) > 0$ for all firms $i \in \ncal$ and all times $t \in \bbt$.
\end{assumption}

To incorporate the possibility of firms defaulting before the terminal time, let $\ncal_0^t(V_{t-1}) \subseteq \ncal_0$ denote the firms that are paying obligations at time $t \in \bbt$ based on the history of wealths up to time $t-1$.  In particular, we will assume that $\ncal_0^0(V_{-1}) := \ncal_0$ and $\ncal_0^{t+1}(V_t) \subseteq \ncal_0^t(V_{t-1})$ for any time $t$ and any wealths process $V$.  That is, all firms are deemed solvent at time $0$ as in Assumption~\ref{ass:initial} and no firm recovers from default.  This notion allows for a consideration in much the same manner as \cite{CC15}.  Mathematically this does not require further consideration than in \cite{BBF18} as $\ncal_0^t$ only depends on the history up to time $t-1$.  With this notation we can define $L_{ij}(t,V_{t-1}) = 0$ for all firms $j \in \ncal_0$ and $i \not\in \ncal_0^t(V_{t-1})$.  With the notion of an auction from \cite{CC15} it will also follow that $L_{ji}(t,V_{t-1}) = 0$ for all firms $j \in \ncal_0$ and $i \not\in \ncal_0^t(V_{t-1})$.
We define the total liabilities and relative liabilities at time $t \in \bbt$ as
\begin{align*}
\bar p_i(t,V_{t-1}) &:= \sum_{j \in \ncal_0} L_{ij}(t,V_{t-1}) + V_i(t-1)^-\\
\pi_{ij}(t,V_{t-1}) &:= \begin{cases} \frac{L_{ij}(t,V_{t-1}) + \pi_{ij}(t-1,V_{t-2})V_i(t-1)^-}{\bar p_i(t,V_{t-1})} &\text{if } \bar p_i(t,V_{t-1}) > 0 \\ \frac{1}{n} &\text{if } \bar p_i(t,V_{t-1}) = 0, \; j \neq i\\ 0 &\text{if } \bar p_i(t,V_{t-1}) = 0, \; j = i \end{cases} \quad \forall i,j \in \ncal_0.
\end{align*}
Then the clearing wealths must satisfy the following fixed point problem in time $t$ wealths:
\begin{equation}\label{eq:EN-discrete}
V(t) = V(t-1)^+ + x(t) + \Pi(t,V_{t-1})^\top \operatorname{diag}(\ind_{\{i \in \ncal_0^t(V_{t-1})\}}) \left[\bar p(t,V_{t-1}) - V(t)^-\right]^+ - \bar p(t,V_{t-1}).
\end{equation}

We proceed to reformulate the problem as in \cite{BBF18}. We consider a process of cash flows $c$ and functional relative exposures $A$.  These we define by 
\begin{align}
\nonumber c(t,V_{t-1}) &:= x(t) + L(t,V_{t-1})^\top \vec{1} - L(t,V_{t-1})\vec{1}\\
\label{eq:discrete-A} 
a_{ij}(t,V_t) &:= \begin{cases} \pi_{ij}(t,V_{t-1}) &\text{if } \bar p_i(t,V_{t-1}) \geq V_i(t)^-,\; i \in \ncal_0^t(V_{t-1})\\ \frac{L_{ij}(t,V_{t-1}) + a_{ij}(t-1,V_{t-1})V_i(t-1)^-}{V_i(t)^-} &\text{else}\end{cases} \quad \forall i,j \in \ncal_0.
\end{align}
That is, we consider $c(t,V_{t-1}) \in \bbr^{n+1}$ to be the vector of book capital levels at time $t$, i.e., the new wealth of each firm assuming all other firms pay in full. We can also consider $c_i(t,V_{t-1})$ to be the \emph{net cash flow} for firm $i$ at time $t$. We define the functional matrix $A: \bbt \times \bbr^{(n+1) \times |\bbt|} \to [0,1]^{(n+1) \times (n+1)}$ to be the relative exposure matrix.  That is, $a_{ij}(t,V_t)V_i(t)^-$ provides the (negative) impact that firm $i$'s losses have on firm $j$'s wealth at time $t \in \bbt$. This is in contrast to $\Pi$, the relative liabilities, in that it endogenously imposes the limited exposures concept.  This equivalent formulation provides mathematical simplicity to the analysis. A broader discussion of this reformulation is provided in \cite{BBF18}.

Thus the fixed point equation reduces to 
\begin{equation}\label{eq:discrete-V}
V(t) = V(t-1) + c(t,V_{t-1}) - A(t,V_t)^\top V(t)^- + A(t-1,V_{t-1})^\top V(t-1)^-.
\end{equation}

 With this setup we now wish to extend the existence and uniqueness results of \cite{EN01} to discrete time.
\begin{corollary}\label{thm:discrete}
Let $(c,L): \bbt \times \bbr^{(n+1) \times |\bbt|} \to \bbr^{n+1} \times \bbr^{(n+1) \times (n+1)}_+$ define a dynamic financial network such that every bank has cash flow at least at the level dictated by nominal interbank liabilities, i.e., $c_i(t,V_{t-1}) \geq \sum_{j \in \ncal} L_{ji}(t,V_{t-1}) - \sum_{j \in \ncal_0} L_{ij}(t,V_{t-1})$ for all times $t \in \bbt$ and all wealth processes $V$, and so that every bank owes to the societal node at all times $t \in \bbt$, i.e., $L_{i0}(t,V_{t-1}) > 0$ for all banks $i \in \ncal$, times $t \in \bbt$, and wealths $V$.  Under Assumption~\ref{ass:initial}, there exists a unique solution of clearing wealths $V: \bbt \to \bbr^{n+1}$ to \eqref{eq:discrete-V}.
\end{corollary}
\begin{proof}
This follows directly from the proof of Theorem 3.2 of \cite{BBF18}.
\end{proof}

With the construction of the existence and uniqueness of the solution, we now want to emphasize the \emph{fictitious default algorithm} from \cite{EN01} to construct this clearing wealths vector over time.  This algorithm is nearly identical to that presented in \cite{BBF18}.  We note that at each time $t$ this algorithm takes at most $n$ iterations as is the case for the fictitious default algorithm originally presented in \cite{EN01}.  Thus with a terminal time $T$, this algorithm will construct the full clearing solution over $\bbt$ in $nT$ iterations.
\begin{alg}\label{alg:discrete}
Under the assumptions of Theorem~\ref{thm:discrete}, the clearing wealths process $V: \bbt \to \bbr^{n+1}$ can be found by the following algorithm.  
Initialize $t = -1$ and $V(-1) \geq 0$ as a given.  Repeat until $t = \max\bbt$:
\begin{enumerate}
\item Increment $t = t+1$.
\item \label{alg:v} Initialize $k = 0$, $V^0 = V(t-1) + c(t,V_{t-1})$, and $D^0 = \emptyset$.  Repeat until convergence:
\begin{enumerate}
\item Increment $k = k+1$;
\item Denote the set of illiquid banks by
$D^k := \left\{i \in \ncal_0^t(V_{t-1}) \; | \; V_i^{k-1} < 0\right\}$.
\item If $D^k = D^{k-1}$ then terminate and set $V(t) = V^{k-1}$.
\item Define the matrix $\Lambda^k \in \{0,1\}^{n \times n}$ so that
$\Lambda_{ij}^k = \begin{cases}1 &\text{if } i = j \in D^k \\ 0 &\text{else}\end{cases}$.
\item \label{alg:vk} Define $V^k = (I - \Pi(t,V_{t-1})^\top \Lambda^k)^{-1}\left(V(t-1) + c(t,V_{t-1}) + A(t-1,V_{t-1})^\top V(t-1)^-\right)$.
\end{enumerate}
\end{enumerate}
\end{alg}
As in \cite{BBF18}, in step~\eqref{alg:vk} of the fictitious default algorithm we are able to replace $A(t,V_t)$ with $\Pi(t,V_{t-1})$.  This is beneficial as it allows us to directly compute $V^k$ without requiring a fixed point problem.  We additionally note that the inclusion of defaulted banks only required the change that the fictitious set of illiquid banks is a subset of $\ncal_0^t(V_{t-1})$ at each time $t$.

\begin{remark}\label{flexibility}
The dynamic framework provides a flexible way to deal with contingent payments. In particular, we can have as many time steps as the number of contingent payment layers in the network. For example, to consider insurance we need to have two time points to incorporate the nominal claims and the insurance claims triggered by the clearing of these nominal claims. For reinsurance markets, we need three time steps, the third one to incorporate the reinsurance claims triggered by the clearing of the insurance claims. We feel this hierarchical resolution of the claims is widely observed in reality.
\end{remark}

\begin{remark}\label{dynamic bankruptcy}
One of the advantages of the dynamic framework is that it provides a natural way to include bankruptcy costs. This is a deviation from the static framework where we might not have existence of solutions for bankruptcy costs. However in the dynamic framework we can always determine the time point when the equity of a bank reaches zero and include the bankruptcy costs for the successive time periods. Hence the solution will exist and be unique.
\end{remark}

\begin{remark}\label{dynamic data dependence}
We can provide much stronger sensitivity results in this case, as compared to the static case. Since in this approach at every time step we get an Eisenberg-Noe system, the sensitivity results are a sequential application of Section 4 of \cite{EN01}.  Directional derivatives of the static Eisenberg-Noe approach have been considered in \cite{LS10,feinstein2017sensitivity}.
\end{remark}

We wish to finish this section by remarking on when the dynamic framework presented herein will provide a clearing solution from the simultaneous claim setting in the prior section.  
\begin{remark}\label{dynamic=static}
In general, the clearing solutions of the simultaneous claims framework will not coincide with the terminal clearing wealths of the dynamic framework.  These notions will, however, coincide if the relative liabilities are kept constant as a function of wealths and time. For a more detailed discussion see Section 5 of~\cite{BBF18}.  Other settings, as evidenced by the examples provided in the next section, may provide sufficient conditions for the dynamic framework to provide a clearing solution from the simultaneous clearing setting.  In particular, this will occur if the contingent payments do not strongly feedback into the network itself, e.g.\ if insurance is owed to an already solvent firm. However, we want to emphasize that the conditions under which the static and dynamic solutions coincide are very restrictive and in general this will not be the case.  This is appropriate given the shortcomings of the static setting as expressed in Section \ref{sec:disadvantages}.
\end{remark}

\subsection{Examples}
We now wish to provide three illustrative examples to demonstrate the value of the discrete time setting as a model over the static setting presented in Section~\ref{sec:insurance}.  These three examples correspond to simple networks in which the static setting has no clearing wealths, has multiple clearing wealths, and has a poor interpretation of the clearing wealths respectively.  We will show that in all three situations the discrete time model presented above provides a unique clearing wealth for which the interpretation of the results is as anticipated.

\begin{example}\label{ex:dynamic-nonexist}
We wish to consider a small network example in which the financial system does not admit a clearing solution in the static setting, but a unique and financially meaningful solution in the dynamic setting.  In this case we will consider a digital CDS.  That is, in the case the CDS is triggered, the payment is a fixed strictly positive value (herein set to be $1$), otherwise it pays out nothing. Immediately we can see that this is not a continuous payout and therefore does not automatically provide a clearing solution in the static setting (see Corollary~\ref{cor:exist-cont}), however we will still need to prove that there does not exist any solution.

Consider the network with $n = 3$ banks, and \emph{without} the societal node, depicted in Figure \ref{fig:nonexist}.  That is, bank 1 begins with $x_1 = 1$, bank 2 with $x_2 = 0$, and bank 3 with $x_3 = 2$ in external assets.  We consider the case in which $L_{12} \equiv 2$ and $L_{23} \equiv 1.5$ are fixed obligations whereas the first bank has purchased a digital credit default swap on the second institution defaulting on its obligations from the third institution, i.e.\ $L_{31}(V) = \ind_{\{V_2 < 0\}}$.  No other exposures exist within this system.  The system of wealths must therefore satisfy
\begin{align*}
V_1 &= 1 + (\ind_{\{V_2 < 0\}} - V_3^-)^+ - 2\\
V_2 &= (2-V_1^-)^+ - 1.5\\
V_3 &= 2 + (1.5 - V_2^-)^+ - \ind_{\{V_2 < 0\}}.
\end{align*}
To show that no clearing solution exists to this system, we will consider the two possible settings: bank 2 is solvent or bank 2 has negative wealth.
\begin{enumerate}
\item Assume bank 2 is solvent, i.e.\ $\ind_{\{V_2 < 0\}} = 0$.  We can compute a unique solution to the clearing wealths $V = (-1,-0.5,3)^\top$.  However, since this violates our assumption that $V_2 \geq 0$, this cannot be a clearing solution to the full problem.
\item Assume bank 2 is insolvent, i.e.\ $\ind_{\{V_2 < 0\}} = 1$.  We can compute a unique solution to the clearing wealths $V = (0,0.5,2.5)^\top$.  However, since this violates our assumption that $V_2 < 0$, this cannot be a clearing solution to the full problem.
\end{enumerate}
As no other possible clearing solutions can exist, it must be the case that there does not exist a clearing solution to this \emph{static} financial system.
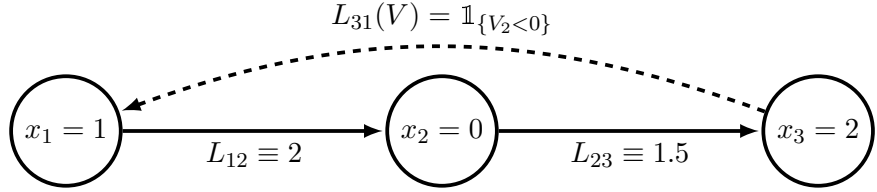
\begin{figure}[h!]
\centering
\begin{tikzpicture}
\tikzset{node style/.style={state, minimum width=0.36in, line width=0.5mm, align=center}}
\node[node style] at (0,0) (x1) {$x_1 = 1$};
\node[node style] at (5,0) (x2) {$x_2 = 0$};
\node[node style] at (10,0) (x3) {$x_3 = 2$};

\draw[every loop, auto=right, line width=0.5mm, >=latex]
(x1) edge node {$L_{12} \equiv 2$} (x2)
(x2) edge node {$L_{23} \equiv 1.5$} (x3)
(x3) edge[bend right=20,dashed] node {$L_{31}(V) = \ind_{\{V_2 < 0\}}$} (x1);
\end{tikzpicture}
\caption{Example~\ref{ex:dynamic-nonexist}: A graphical representation of the network model with 3 banks which has no clearing in a static setting.}
\label{fig:nonexist}
\end{figure}

Now we wish to consider the same example but in the discrete time framework with $\bbt = \{0,1\}$.  Here we will consider all possible divisions of the external assets over the two time points.  Formally, define $x^{\epsilon}(0) = (\epsilon,0,1)^\top$ and $x^{\epsilon}(1) = (1-\epsilon,0,1)^\top$ for any $\epsilon \in [0,1]$.  Note that, by the topology of this network, it is a regular network (as defined in \cite{EN01}) for any choice of $\epsilon \in [0,1]$ as required by Assumption~\ref{ass:initial}.  In any scenario, define $L_{12}(0) = 2$ and $L_{23}(0) = 1.5$ with no other obligations at time $0$.  The only new obligation owed at time 1 is the contingent payment from bank 3 to 1, i.e., $L_{31}(1,V_0) = \ind_{\{V_2(0) < 0\}}$ with no other new obligations at time $1$.  Further, all scenarios will be assumed to start from zero wealths (thus satisfying Assumption~\ref{ass:initial}).  We can easily compute the \emph{unique} clearing wealths under $x^{\epsilon}$ (assuming no firms are removed from the system) as $V^{\epsilon}(0) = (\epsilon-2,\epsilon-1.5,\epsilon+1)^\top$ and (noting that $V_2^{\epsilon}(0) < 0$ for any $\epsilon \in [0,1]$) $V^{\epsilon}(1) = (0,0.5,2.5)^\top$.  We note that this clearing solution is identical to the proposed \emph{static} wealths under the assumption that bank 2 is insolvent.  Additionally, the final wealths are independent of the choice of $\epsilon$.
\end{example}

\begin{example}\label{ex:dynamic-nonunique}
Consider again Example~\ref{ex:nonunique} with three banks.  In the static solution this was encoded by the parameters: external assets of $x = (0,3/16,0)^\top$ and sparse liabilities provided by $L_{23} = L_{32} \equiv 1$ and $L_{21}(V) = V_3^-$.  Two clearing solutions existed, $V^* = (0,3/16,0)^\top$ and $V^* = (3/16,-21/16,-3/4)^\top$.  

Now we wish to consider the same example but in the discrete time framework with $\bbt = \{0,1\}$.  Here we will consider all possible divisions of the external assets over the two time points.  Formally, define $x^{\epsilon}(0) = (0,\epsilon,0)^\top$ and $x^{\epsilon}(1) = (0,3/16-\epsilon,0)^\top$ for any $\epsilon \in (0,3/16]$ to guarantee the uniqueness of the clearing solutions as a regular network from \cite{EN01} (and as required from Assumption~\ref{ass:initial} and since no societal node is included in this example).  In any scenario, define $L_{23}(0) = L_{32}(0) \equiv 1$ with no other obligations at time $0$.  The only new obligation owed at time 1 is the contingent payment from bank 2 to 1, i.e., $L_{21}(1,V_0) = V_3(0)^-$ with no other new obligations at time $1$.  Further, all scenarios will be assumed to start from zero wealths (thus satisfying Assumption~\ref{ass:initial}).  We can easily compute the \emph{unique} clearing wealths under $x^{\epsilon}$ as $V^{\epsilon}(0) = (0,\epsilon,0)^\top$ and $V^{\epsilon}(1) = (0,3/16,0)^\top$.  We note that this clearing solution is identical to the first clearing wealths solution of the static system and is independent of the choice of $\epsilon$.
\end{example}

\begin{example}\label{ex:interpretation}
Finally, we want to consider a simple financial system to demonstrate the issues discussed in Section~\ref{sec:disadvantages} surrounding the static framework.  We will then use this same network in the discrete time framework to find a unique, financially meaningful, clearing solution.  To do so, consider a bank who takes out an insurance payment on its own losses.  As discussed in Section~\ref{sec:disadvantages}, while the insured bank may, rightly, assume that their total losses will be made whole, in a static setting this will not happen.  However, in the dynamic framework this does occur appropriately.

Consider the network with $n = 3$ banks, and \emph{without} the societal node, depicted in Figure \ref{fig:interpretation}.  That is, bank 1 begins with $x_1 = 1$, bank 2 with $x_2 = 0$, and bank 3 with $x_3 = 2$ in external assets.  We consider the case in which $L_{12} \equiv 2$ and $L_{23} \equiv 1.5$ are fixed obligations whereas the first bank has purchased insurance on their own losses from the third institution, i.e.\ $L_{31}(V) = V_1^-$.  No other exposures exist within this system.  The system of wealths must therefore satisfy
\begin{align*}
V_1 &= 1 + (V_1^- - V_3^-)^+ - 2\\
V_2 &= (2-V_1^-)^+ - 1.5\\
V_3 &= 2 + (1.5 - V_2^-)^+ - V_1^-.
\end{align*}
Without the insurance payment, the first bank will default with wealths of $V = (-1,-0.5,3)^\top$.  However, if the insurance is paid out in full then the first bank is made whole and the resultant wealths are $V = (0,0.5,2.5)$.  In this case, the first bank does not default, which raises the question whether any insurance payment is to be made at all.  Neither of these are clearing solutions as the system would infinitely cycle between needing insurance or not.  The clearing wealths, instead, are given by $V = (-0.5,0,3)^\top$.  That is, bank 1 will have a shortfall midway between its wealth with and without the insurance being paid.  This, though, is \emph{not} the notion that a firm purchasing insurance would expect as it cannot make them whole.
\begin{figure}[h!]
\centering
\begin{tikzpicture}
\tikzset{node style/.style={state, minimum width=0.36in, line width=0.5mm, align=center}}
\node[node style] at (0,0) (x1) {$x_1 = 1$};
\node[node style] at (5,0) (x2) {$x_2 = 0$};
\node[node style] at (10,0) (x3) {$x_3 = 2$};

\draw[every loop, auto=right, line width=0.5mm, >=latex]
(x1) edge node {$L_{12} \equiv 2$} (x2)
(x2) edge node {$L_{23} \equiv 1.5$} (x3)
(x3) edge[bend right=20,dashed] node {$L_{31}(V) = V_1^-$} (x1);
\end{tikzpicture}
\caption{Example~\ref{ex:interpretation}: A graphical representation of the network model with 3 banks which has poor interpretation in a static setting.}
\label{fig:interpretation}
\end{figure}
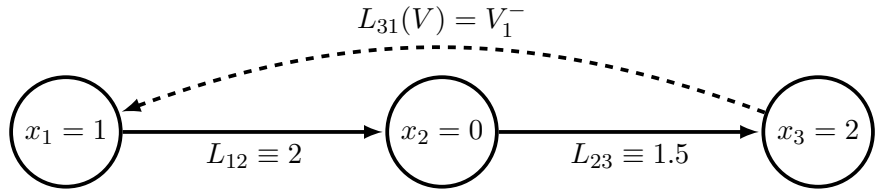

Now we wish to consider the same example but in the discrete time framework with $\bbt = \{0,1\}$.  Here we will consider all possible divisions of the external assets over the two time points.  Formally, define $x^{\epsilon}(0) = (\epsilon,0,1)^\top$ and $x^{\epsilon}(1) = (1-\epsilon,0,1)^\top$ for any $\epsilon \in [0,1]$.  Note that, by the topology of this network, it is a regular network (as defined in \cite{EN01}) for any choice of $\epsilon \in [0,1]$ as required by Assumption~\ref{ass:initial}.  In any scenario, define $L_{12}(0) = 2$ and $L_{23}(0) = 1.5$ with no other obligations at time $0$.  The only new obligation owed at time 1 is the contingent payment from bank 3 to 1, i.e., $L_{31}(1,V_0) = V_1(0)^-$ with no other new obligations at time $1$.  Further, all scenarios will be assumed to start from zero wealths (thus satisfying Assumption~\ref{ass:initial}).  We can easily compute the \emph{unique} clearing wealths under $x^{\epsilon}$ (assuming no firms are removed from the system) as $V^{\epsilon}(0) = (\epsilon-2,\epsilon-1.5,\epsilon+1)^\top$ and (noting that $V_1^{\epsilon}(0) < 0$ for any $\epsilon \in [0,1]$) $V^{\epsilon}(1) = (1-\epsilon,0.5,1.5+\epsilon)^\top$.  We note that, in the case that $\epsilon = 1$, this clearing solution is identical to the proposed \emph{static} wealths when the insurance is paid in full.  As opposed to the prior examples, here the final wealths are a function of $\epsilon$.
\end{example}

\section{Conclusion}\label{sec:conclusion}

In this paper we consider an extension of the network model of \cite{EN01} to include contingent payments viz.\ insurance and CDSs with endogenous reference entities. We first study these contingent payments in a static, simultaneous claims, framework and develop conditions to provide existence and uniqueness of the clearing wealths.  Further, sensitivity analysis and financial implications are considered in this setting. We find that the static framework is suitable only for a certain class of networks and we cannot guarantee the existence of a clearing solution beyond these systems. Indeed the problem often becomes ill-defined from a financial standpoint. Hence we introduce the dynamic framework and show that we can get existence and uniqueness under very mild assumptions. Further we show that the problems which could not be solved in the simultaneous claims framework can be studied with this dynamic approach.

A clear extension of this model would be to include illiquid assets as discussed in, e.g., \cite{AFM16,feinstein2015illiquid,feinstein2016leverage} along with financial derivatives on these illiquid assets, i.e., options. These derivatives fall under the general class of contingent payments and can be used a tool for either hedging (insurance) or speculation. Due to the possibility of speculation, in such a setting a firm may have incentives to attempt to precipitate a fire sale and collect profit from the derivatives. 

\appendix
\section{Proof of Proposition~\ref{prop:continuous}}\label{sec:discrete-proofs}

\begin{proof}
Firstly, as in~\eqref{eq:EN-insurance}, the clearing wealths as a function of initial endowments are defined by
\[V(x) = x + \Pi(V(x))^\top [\bar p(V(x)) - V(x)^-]^+ - \bar p(V(x)).\]
We will prove continuity by utilizing the closed graph theorem (see, e.g., \cite[Theorem 2.58]{AB07}) noting that Proposition~\ref{prop:V-bounded} provides us with the condition that the clearing wealths map into a compact set.  Theorem 4 of \cite{MR94} immediately provides the monotonicity of the clearing wealths.

Fix $x \in \bbr^{n+1}_+$ and let $\xcal = x + [-1,1]^{n+1}$ be a closed compact neighborhood of $x$ in the full Euclidean space $\bbr^{n+1}$.  Then we can define $V^x: \xcal \to \bbr^{n+1}$ as the restriction (and possible expansion to negative terms) of the domain of $V$ to $\xcal$.  The graph of $V^x$ is given by:
\[\operatorname{graph}V^x = \left\{(\hat x,\hat V) \in \xcal \times \prod_{i \in \ncal_0} [x_i-1-\sum_{j \in \ncal_0} \bar L_{ij} , x_i+1+\sum_{j \in \ncal} \bar L_{ji}] \; | \; \hat V = \hat x + \Pi(\hat V)^\top [\bar p(\hat V) - \hat V^-]^+ - \bar p(\hat V)\right\}.\]
To see that $\operatorname{graph}V^x$ is closed let $(\hat x^k,\hat V^k)_{k \in \bbn} \subseteq \operatorname{graph}V^x \to (\hat x,\hat V)$, then immediately
\[\hat V = \lim_{k \to \infty} \hat V^k = \lim_{k \to \infty} \left[\hat x^k + \Pi(\hat V^k)^\top [\bar p(\hat V^k) - (\hat V^k)^-]^+ - \bar p(\hat V^k)\right] = \hat x + \Pi(\hat V)^\top [\bar p(\hat V) - \hat V^-]^+ - \bar p(\hat V)\]
by continuity of the nominal liabilities matrix $L$.  Therefore by the closed graph theorem we immediately recover that $V^x$ is continuous for any $x \in \bbr^{n+1}_+$, which implies that $V$ is continuous at any $x$ as well and thus $V: \bbr^{n+1}_+ \to \bbr^{n+1}$ is a continuous mapping.
\end{proof}

\bibliographystyle{plain}
\bibliography{bibtex25}

\begin{thebibliography}{10}

\bibitem{AB2014ccp}
Viral Acharya and Alberto Bisin.
\newblock Counterparty risk externality:centralized versus over-the-counter
  markets.
\newblock {\em Journal of Economic Theory}, 149:153--182, 2014.

\bibitem{AB07}
Charalambos~D. Aliprantis and Kim~C. Border.
\newblock {\em Infinite Dimensional Analysis: A Hitchhiker's Guide}.
\newblock Springer, 2007.

\bibitem{AFM13}
Hamed Amini, Damir Filipovi\'{c}, and Andreea Minca.
\newblock Systemic risk with central counterparty clearing.
\newblock Swiss {Finance} {Institute} {Research} {Paper} {No.} 13-34, Swiss
  Finance Institute, 2015.

\bibitem{AFM16}
Hamed Amini, Damir Filipovi\'{c}, and Andreea Minca.
\newblock Uniqueness of equilibrium in a payment system with liquidation costs.
\newblock {\em Operations Research Letters}, 44(1):1--5, 2016.

\bibitem{AW_15}
Kerstin Awiszus and Stefan Weber.
\newblock The joint impact of bankruptcy costs, cross-holdings and fire sales
  on systemic risk in financial networks.
\newblock {\em Probability, Uncertainty and Quantitative Risk}, 2(9):1--38,
  2017.

\bibitem{BBF18}
Tathagata Banerjee, Alex Bernstein, and Zachary Feinstein.
\newblock Dynamic clearing and contagion in financial networks.
\newblock 2018.
\newblock Working paper.

\bibitem{barucca2016valuation}
Paolo Barucca, Marco Bardoscia, Fabio Caccioli, Marco D'Errico, Gabriele
  Visentin, Stefano Battiston, and Guido Caldarelli.
\newblock Network valuation in financial systems.
\newblock 2016.
\newblock Working paper.

\bibitem{BS12}
Jose Blanchet and Yixi Shi.
\newblock Stochastic risk networks: Modeling, analysis and efficient monte
  carlo.
\newblock 2012.
\newblock Working paper.

\bibitem{CC15}
Agostino Capponi and Peng-Chu Chen.
\newblock Systemic risk mitigation in financial networks.
\newblock {\em Journal of Economic Dynamics and Control}, 58:152--166, 2015.

\bibitem{CCY16}
Agostino Capponi, Peng-Chu Chen, and David~D. Yao.
\newblock Liability concentration and systemic losses in financial networks.
\newblock {\em Operations Research}, 64(5):1121--1134, 2016.

\bibitem{CLY14}
Nan Chen, Xin Liu, and David~D. Yao.
\newblock An optimization view of financial systemic risk modeling: The network
  effect and the market liquidity effect.
\newblock {\em Operations Research}, 64(5):1089--1108, 2016.

\bibitem{CFS05}
Rodrigo Cifuentes, Hyun~Song Shin, and Gianluigi Ferrucci.
\newblock Liquidity risk and contagion.
\newblock {\em Journal of the European Economic Association}, 3(2-3):556--566,
  2005.

\bibitem{cont2015ccp}
Rama Cont.
\newblock The end of the waterfall: Default resources of central
  counterparties.
\newblock {\em Journal of Risk Management in Financial Institutions},
  8(4):365--389, 2015.

\bibitem{cont2013ito}
Rama Cont and David-Antoine Fourni\'e.
\newblock Functional {I}t\^o calculus and stochastic integral representation of
  martingales.
\newblock {\em The Annals of Probability}, 41(1):109--133, 2013.

\bibitem{CM2016}
Rama Cont and Andreea Minca.
\newblock Credit default swaps and systemic risk.
\newblock {\em Annals of Operations Research}, 247:523--547, 2016.

\bibitem{CMS10}
Rama Cont, Amal Moussa, and Edson Bastos~e Santos.
\newblock Network structure and systemic risk in banking systems.
\newblock In {\em Handbook on Systemic Risk}, pages 327--368. Cambridge
  University Press, 2013.

\bibitem{EN01}
Larry Eisenberg and Thomas~H. Noe.
\newblock Systemic risk in financial systems.
\newblock {\em Management Science}, 47(2):236--249, 2001.

\bibitem{EGJ14}
Matthew Elliott, Benjamin Golub, and Matthew~O. Jackson.
\newblock Financial networks and contagion.
\newblock {\em American Economic Review}, 104(10):3115--3153, 2014.

\bibitem{E07}
Helmut Elsinger.
\newblock Financial networks, cross holdings, and limited liability.
\newblock {\em {\"{O}sterrei}chische Nationalbank (Austrian Central Bank)},
  156, 2009.

\bibitem{ELS06}
Helmut Elsinger, Alfred Lehar, and Martin Summer.
\newblock Risk assessment for banking systems.
\newblock {\em Management Science}, 52(9):1301--1314, 2006.

\bibitem{feinstein2015illiquid}
Zachary Feinstein.
\newblock Financial contagion and asset liquidation strategies.
\newblock {\em Operations Research Letters}, 45(2):109--114, 2017.

\bibitem{feinstein2017multilayer}
Zachary Feinstein.
\newblock Obligations with physical delivery in a multi-layered financial
  network.
\newblock 2018.
\newblock Working paper.

\bibitem{feinstein2016leverage}
Zachary Feinstein and Fatena El-Masri.
\newblock The effects of leverage requirements and fire sales on financial
  contagion via asset liquidation strategies in financial networks.
\newblock {\em Statistics and Risk Modeling}, 2017.

\bibitem{feinstein2017sensitivity}
Zachary Feinstein, Weijie Pang, Birgit Rudloff, Eric Schaanning, Stephan Sturm,
  and Mackenzie Wildman.
\newblock Sensitivity of the {E}isenberg--{N}oe clearing vector to individual
  interbank liabilities.
\newblock {\em SIAM Journal on Financial Mathematics}, 2018.
\newblock To appear.

\bibitem{feinstein2014measures}
Zachary Feinstein, Birgit Rudloff, and Stefan Weber.
\newblock Measures of systemic risk.
\newblock {\em SIAM Journal on Financial Mathematics}, 8(1):672--708, 2017.

\bibitem{GK10}
Prasanna Gai and Sujit Kapadia.
\newblock Contagion in financial networks.
\newblock Bank of England Working Papers 383, Bank of England, 2010.

\bibitem{GY14}
Paul Glasserman and H.~Peyton Young.
\newblock How likely is contagion in financial networks?
\newblock {\em Journal of Banking and Finance}, 50:383--399, 2015.

\bibitem{HS12}
Sebastian Heise and Reimer K\"{u}hn.
\newblock Derivatives and credit contagion in interconnected networks.
\newblock {\em The European Physical Journal B}, 85(4):115, 2012.

\bibitem{minca2018reinsurance}
Ariah Klages-Mundt and Andreea Minca.
\newblock Cascading losses in reinsurance networks.
\newblock 2018.
\newblock Working paper.

\bibitem{KV16}
Michael Kusnetsov and Luitgard A.~M. Veraart.
\newblock Interbank clearing in financial networks with multiple maturities.
\newblock 2018.
\newblock Working paper.

\bibitem{LPT17}
Matt Leduc, Sebastian Poledna, and Stefan Thurner.
\newblock Systemic risk management in financial networks with credit default
  swaps.
\newblock {\em Journal of Network Theory in Finance}, 3(3):19--39, 2017.

\bibitem{LS10}
Ming Liu and Jeremy Staum.
\newblock Sensitivity analysis of the {E}isenberg-{N}oe model of contagion.
\newblock {\em Operations Research Letters}, 35(5):489--491, 2010.

\bibitem{SGGS10}
Sheri~M Markose, Simone Giansante, Mateusz Gatkowski, and Ali~Rais Shaghaghi.
\newblock Too interconnected to fail: Financial contagion and systemic risk in
  network model of {CDS} and other credit enhancement obligations of us banks.
\newblock Technical Report DP 683, Economics Department, University of Essex,
  2010.

\bibitem{MR94}
Paul Milgrom and John Roberts.
\newblock Comparing equilibria.
\newblock {\em American Economic Review}, 84(3):441--459, 1994.

\bibitem{Murphy2012ccp}
David Murphy.
\newblock The systemic risk of otc derivatives central clearing.
\newblock {\em Journal of Risk Management in Financial Institutions},
  5(3):319--334, 2012.

\bibitem{NYYA07}
Erland Nier, Jing Yang, Tanju Yorulmazer, and Amadeo Alentorn.
\newblock Network models and financial stability.
\newblock {\em Journal of Economic Dynamics and Control}, 31(6):2033--2060,
  2007.

\bibitem{PRY16}
Mark~E. Paddrik, Sriram Rajan, and Peyton Young.
\newblock Contagion in the {CDS} market.
\newblock OFR Working Paper 16-12, Office of Financial Research, 2016.

\bibitem{PCB14}
Michelangelo Puliga, Guido Caldarelli, and Stefano Battiston.
\newblock Credit default swaps networks and systemic risk.
\newblock {\em Scientific Reports}, 4, 2014.

\bibitem{RV13}
L.~C.~G. Rogers and L.~A.~M. Veraart.
\newblock Failure and rescue in an interbank network.
\newblock {\em Management Science}, 59(4):882--898, 2013.

\bibitem{SSB17}
Steffen Schuldenzucker, Sven Seuken, and Stefano Battiston.
\newblock Default ambiguity: Credit default swaps create new systemic risks in
  financial networks.
\newblock 2017.
\newblock Working Paper.

\bibitem{SSB16b}
Steffen Schuldenzucker, Sven Seuken, and Stefano Battiston.
\newblock {Finding Clearing Payments in Financial Networks with Credit Default
  Swaps is PPAD-complete}.
\newblock In Christos~H. Papadimitriou, editor, {\em 8th Innovations in
  Theoretical Computer Science Conference (ITCS 2017)}, volume~67 of {\em
  Leibniz International Proceedings in Informatics (LIPIcs)}, pages
  32:1--32:20, Dagstuhl, Germany, 2017. Schloss Dagstuhl--Leibniz-Zentrum fuer
  Informatik.

\bibitem{U11}
Christian Upper.
\newblock Simulation methods to assess the danger of contagion in interbank
  markets.
\newblock {\em Journal of Financial Stability}, 7(3):111--125, 2011.

\bibitem{veraart2017distress}
Luitgard~A.M. Veraart.
\newblock Distress and default contagion in financial networks.
\newblock 2018.
\newblock Working paper.

\end{thebibliography}

\end{document}